\documentclass[11pt]{article}
\usepackage{verbatim,amsmath,amsthm,eufrak}
\usepackage{cite} 
\usepackage{mathrsfs}
\usepackage{makeidx}
\usepackage{yhmath}
\usepackage{enumitem}
\usepackage{color,xcolor}
\usepackage[hypertexnames=false,hyperfootnotes=false,colorlinks=true,linkcolor=blue,%
citecolor=purple,filecolor=magenta,urlcolor=cyan,unicode,linktocpage=true,pagebackref=false]{hyperref}
\usepackage{scalerel,stackengine}

\renewcommand{\arraystretch}{1.2}
\makeatletter
\newdimen\normalarrayskip              
\newdimen\minarrayskip                 
\normalarrayskip\baselineskip
\minarrayskip\jot
\newif\ifold             \oldtrue            \def\new{\oldfalse}
\def\arraymode{\ifold\relax\else\displaystyle\fi} 
\def\eqnumphantom{\phantom{(\theequation)}}     
\def\@arrayskip{\ifold\baselineskip\z@\lineskip\z@
     \else
     \baselineskip\minarrayskip\lineskip2\minarrayskip\fi}
\def\@arrayclassz{\ifcase \@lastchclass \@acolampacol \or
\@ampacol \or \or \or \@addamp \or
   \@acolampacol \or \@firstampfalse \@acol \fi
\edef\@preamble{\@preamble
  \ifcase \@chnum
     \hfil$\relax\arraymode\@sharp$\hfil
     \or $\relax\arraymode\@sharp$\hfil
     \or \hfil$\relax\arraymode\@sharp$\fi}}
\def\@array[#1]#2{\setbox\@arstrutbox=\hbox{\vrule
     height\arraystretch \ht\strutbox
     depth\arraystretch \dp\strutbox
     width\z@}\@mkpream{#2}\edef\@preamble{\halign
\noexpand\@halignto
\bgroup \tabskip\z@ \@arstrut \@preamble \tabskip\z@ \cr}%
\let\@startpbox\@@startpbox \let\@endpbox\@@endpbox
  \if #1t\vtop \else \if#1b\vbox \else \vcenter \fi\fi
  \bgroup \let\par\relax
  \let\@sharp##\let\protect\relax
  \@arrayskip\@preamble}
\def\eqnarray{\stepcounter{equation}%
              \let\@currentlabel=\theequation
              \global\@eqnswtrue
              \global\@eqcnt\z@
              \tabskip\@centering
              \let\\=\@eqncr
 \halign to \displaywidth\bgroup
    \eqnumphantom\@eqnsel\hskip\@centering
    $\displaystyle \tabskip\z@ {##}$%
    \global\@eqcnt\@ne \hskip 2\arraycolsep
         $\displaystyle\arraymode{##}$\hfil
    \global\@eqcnt\tw@ \hskip 2\arraycolsep
         $\displaystyle\tabskip\z@{##}$\hfil
         \tabskip\@centering
    &{##}\tabskip\z@\cr}
\begingroup\ifx\undefined\newsymbol \else\def\input#1 {\endgroup}\fi

\newcounter{app}

\def\app{\setcounter{equation}{0}
\def\theequation{A\Roman{app}.\arabic{equation}}\par
   \addvspace{4ex}
   \@afterindentfalse
  \secdef\@app\@dapp}
\newcommand\@app{\@startsection {app}{1}{0ex}%
                                   {-3.5ex \@plus -1ex \@minus -.2ex}%
                                   {2.3ex \@plus.2ex}%
                                   {\normalfont\Large\bf}}

\def\@dapp#1{%
{\parindent \z@ \raggedright  \bf #1}\par\nobreak}
\def\l@app#1#2{\ifnum \c@tocdepth >\z@
    \addpenalty\@secpenalty
    \addvspace{1.0em \@plus\p@}%
    \setlength\@tempdima{8.5em}%
    \begingroup
      \parindent \z@ \rightskip \@pnumwidth
      \parfillskip -\@pnumwidth
      \leavevmode \bfseries
      \advance\leftskip\@tempdima
      \hskip -\leftskip
      #1\nobreak\hfil \nobreak\hb@xt@\@pnumwidth{\hss #2}\par
    \endgroup\fi}
\newcounter{sapp}[app]

\def\sapp{\def\theequation{A\arabic{app}.\arabic{equation}}\par
   \@afterindentfalse
  \secdef\@sapp\@dsapp}
\newcommand\@sapp{\@startsection{sapp}{2}{\z@}%
                                     {-3.25ex\@plus -1ex \@minus -.2ex}%
                                     {1.5ex \@plus .2ex}%
                                     {\normalfont\large\bfseries}}

\def\@dsapp#1{%
{\parindent \z@ \raggedright  \bf #1}\par\nobreak}
\newcommand{\l@sapp}{\@dottedtocline{2}{1.5em}{3em}}
\def\draft{\oddsidemargin -.5truein
        \def\@oddfoot{\sl preliminary draft \hfil
        \rm\thepage\hfil\sl\today\quad\militarytime}
        \let\@evenfoot\@oddfoot \overfullrule 3pt
        \let\label=\draftlabel
        \let\marginnote=\draftmarginnote
   \def\@eqnnum{(\theequation)\rlap{\kern\marginparsep\tt\@eqnlabel}%
\global\let\@eqnlabel\@vacuum}  }

\def\be{\begin{eqnarray}}
\def\ee{\end{eqnarray}}

\def\p{\partial}

\def\beq{\begin{equation}}
\def\eeq{\end{equation}}
\def\ba{\beq\new\begin{array}{c}}
\def\ea{\end{array}\eeq}
\def\be{\ba}
\def\ee{\ea}

\def\Tr{{\rm Tr}\,}
\def\dim{{\rm dim}\,}
\def\res{{\rm res}\,}
\def\vir{{\rm vir}\,}
\def\ev{{\rm ev}\,}
\def\diag{{\rm diag}\,}
\def\sppan{{\rm span}\,}
\def\ev{{\rm ev}}
\def\psistar{\psi^{*}}
\newfont{\Bbbb}{msbm7 scaled 1\@ptsize00}
\newcommand{\ZZ}{\mathbb{Z}}
\newcommand{\z}{\raise-1pt\hbox{$\mbox{\Bbbb Z}$}}
\newcommand{\cc}{\raise-1pt\hbox{$\mbox{\Bbbb C}$}}
\newcommand{\rr}{\raise-1pt\hbox{$\mbox{\Bbbb R}$}}

\DeclareMathOperator{\GL}{GL}
\DeclareMathOperator{\gl}{gl}

\def\rbr{\right >}

\def\lvac{\left <0\right |}
\def\rvac{\left |0\right >}
\def\lvacn{\left <n\right |}
\def\rvacn{\left |n\right >}

\def\normord{ {\scriptstyle {{\bullet}\atop{\bullet}}} }

\newfont{\alef}{msbm10 at 11pt}
\newfont {\goth}{eufm10 at 11pt}
\def\mathbb#1{\hbox{{\alef #1}}}

\let\@@savethanks\thanks
\def\thanks#1{\gdef\thefootnote{\alph{footnote}}\@@savethanks{#1}}

\newtheorem{theorem}{Theorem}
\newtheorem{lemma}{Lemma}[section]
\newtheorem{proposition}[lemma]{Proposition}
\newtheorem{corollary}[lemma]{Corollary}
\newtheorem{remark}{Remark}[section]

\newtheorem*{theorem*}{Theorem}

\numberwithin{equation}{section}

\makeatletter
\g@addto@macro \normalsize {%
 \setlength\abovedisplayskip{14pt plus 3pt minus 3pt}%
 \setlength\belowdisplayskip{14pt plus 3pt minus 3pt}%
  \setlength\abovedisplayshortskip{11pt plus 3pt minus 3pt}%
 \setlength\belowdisplayshortskip{11pt plus 3pt minus 3pt}%
}
\makeatother

\bigskip
\bigskip

\title{
\bigskip
{\bf Matrix model for the stationary sector of Gromov-Witten theory of ${\bf P}^1$} \vspace{.5cm}}
\author{{\bf Alexander Alexandrov}\thanks{E-mail:  {\tt alexandrovsash at gmail.com}}
\date{ } \\
{\small {\it Center for Geometry and Physics, Institute for Basic Science (IBS), Pohang 37673, Korea}}\\
}

\begin{document}

\setcounter{footnote}{0}

\setcounter{tocdepth}{3}

\maketitle

\vspace{-8.0cm}

\begin{center}
\end{center}

\vspace{6.5cm}
\begin{abstract} 
In this paper we investigate the tau-functions for the stationary sector of Gromov-Witten theory of the complex projective line and its version, relative to one point. In particular, we construct the integral representation for the points of the Sato Grassmannians, Kac-Schwarz operators, and quantum spectral curves. This allows us to derive the matrix models. 
\end{abstract}
\bigskip


\bigskip

\newpage

\tableofcontents

\def\thefootnote{\arabic{footnote}}
\section{Introduction}
\setcounter{equation}{0}

Matrix models play an important role in modern enumerative geometry. They are closely related to its other ingredients, including integrable systems, quantum spectral curves  and Chekhov-Eynard-Orantin topological recursion.
The main goal of this paper is to construct matrix models for the generating functions of the stationary sector of Gromov-Witten theory of ${\bf P}^1$ and its relative version. Our construction is based on the known relation between these generating functions of the Gromov-Witten invariants, tau-functions of integrable hierarchies, and free fermions obtained by Okounkov and Pandharipande \cite{OP,OP2}.

Let $\overline{\mathcal M}_{g,n} ({\bf P}^1,d)$ be the moduli space of  $n$-pointed genus $g$ stable maps to ${\bf P}^1$ of degree $d$, $f:(\Sigma, p_1,\dots,p_n) \to {\bf P}^1$. The virtual dimension of  $\overline{\mathcal M}_{g,n} ({\bf P}^1,d)$ is $2g-2+n+2d$. 
Let $\mathcal{L}_i$ be the line bundle on $\overline{\mathcal M}_{g,n} ({\bf P}^1,d)$, whose fiber 
 is the cotangent line  at the $i$th marked point, and $\psi_i$ is the first Chern class of $\mathcal{L}_i$. Consider $\omega\in H^2({\mathbf P}^1,{\mathbb Q})$ be the Poincar{\'e} dual of the point class. Together with $1\in H^0({\mathbf P}^1,{\mathbb Q})$ they constitute the standard basis of $ H^*({\mathbf P}^1,{\mathbb Q})$. Let also $\gamma_i \in \{1, \omega\}$. The connected Gromov-Witten invariants of ${\bf P}^1$ are the integrals
 \be
\left< \prod_{i=1}^n \tau_{k_i}(\gamma_i) \right>_{g,d}:=\int_{[\overline{\mathcal M}_{g,n} ({\bf P}^1,d)]^\vir } \prod_{i=1}^n \psi_i^{k_i} \ev_i^*(\gamma_i),
\ee
where $\ev_i$ is the evaluation map,
defined by evaluating a stable map $f:(\Sigma, p_1,\dots,p_n) \to {\bf P}^1$
at the $i$th marked point, $\ev_i^*(\gamma_i)\in H^*(\overline{\mathcal M}_{g,n},{\mathbb Q})$. These invariants vanish unless $\sum_{i=1}^n\left(\dim (\gamma_i)+k_i\right)=2g-2+n+2d$. Consider the generating function of the Gromov-Witten invariants of ${\bf P}^1$
\be\label{neqpf}
{Z}^*({\bf t^\omega},{\bf t}^1)=\exp\left(\sum_{g=0}^\infty\sum_{d=0}^\infty \hbar^{2g-2} q^d\left<\exp
\left(\sum_{k=0}^\infty t^\omega_k \tau_k(\omega)+t^1_k\tau_k(1) \right)\right>_{g,d}\right).
\ee
Here powers of the parameters $q$ and $\hbar$ keep track of the map's degree and the genus of the curve respectively. 

Relation of this generating function to a proper version of the Toda integrable hierarchy, known as Toda conjecture, was formulated in \cite{Eguchi1,Eguchi2r}. This version of the integrable hierarchy, now known as the {\em extended Toda hierarchy}, was defined and investigated in \cite{Getz0,Zhang,CDZ}. Toda conjecture was discussed in \cite{Pand,OkunToda,LMN}, and proved in \cite{DZ,OP2}.

The {\em stationary sector} of the Gromov-Witten theory is formed by intersection numbers of $\tau_{k}(\omega)$'s, the descendants of $\omega$.
In this paper we consider the {\em extended stationary Gromov-Witten generating function}, which includes the dependence on $t_0^1$:
\be\label{tau11}
\tau({\bf t^\omega},t^1_0):={Z}^*({\bf t^\omega},{\bf t}^1)|_{t^1_{\geq1}=0}.
\ee
This is a tau-function of the MKP hierarchy.

Gromov-Witten theory on ${\bf P}^1$ has a natural equivariant deformation. This deformation can be described by the free fermions in a particularly nice way, derived by Okounkov and Pandharipande \cite{OP}.
 This free fermion description allows us to describe explicitly a $\GL(\infty)$ group element for the nonequivariant MKP tau-function (\ref{tau11}), and to construct the integral expressions for the basis vectors of the Sato Grassmannian point. Then, with the standard methods of matrix models, we prove 
\begin{theorem}\label{T_1}
The extended stationary Gromov-Witten generating function of ${\bf P}^1$ is given by the asymptotic expansion of the  matrix integral
\be\label{Int1}
\tau({\bf t}^\omega,t_0^1)=\frac{e^{\frac{1}{\hbar}\Tr\left((t_0^1-\Lambda)\log\Lambda+\Lambda\right)+\frac{q}{\hbar^2}}}{\hbar^\frac{N^2}{2}}\int_{{\mathcal H}_N} \left[d \mu ({Y})\right] e^{\frac{1}{\hbar}\Tr\left(Y\Lambda -e^Y+ q e^{-Y}+\left(N\hbar/2-t^1_0\right)Y\right)},
\ee
where
\be
t_k^\omega = \hbar\, k!\, Tr \Lambda^{-k-1}.
\ee
\end{theorem}
We also prove that a simple deformation of this matrix model describes a stationary sector of the Gromov-Witten theory of ${\mathbf P}^1$ relative to one point (Theorem \ref{T2}).

Matrix integral (\ref{Int1}) belongs to a family of the generalized Kontsevich models (GKM) \cite{Konts,KMMM,KMMM1,Mortoda,Adler,IZ}. This family is believed to capture the fundamental properties of two-dimensional topological gravity and it is naturally related to the modified KP hierarchy. The model (\ref{Int1}) for the extended stationary Gromov-Witten theory of ${\bf P}^1$ is similar to the original Kontsevich matrix integral \cite{Konts} for the intersection theory on the moduli spaces of the Riemann surfaces. 
The representatives of GKM describe many enumerative geometry tau-functions. In particular, the model (\ref{Int1}) is similar to a model derived earlier for Hurwitz numbers \cite{MSh,MMRP}. It is not very surprising, because the stationary Gromov-Witten invariants are closely related to Hurwitz numbers \cite{OkunToda, Pand,OP2}.

We expect that our methods should help to investigate other models 
of enumerative geometry related to ${\bf P}^1$, and to construct corresponding matrix models. In particular, we expect that our approach should work for the full generating functions of the Gromov-Witten theory, both in equivariant and nonequivariant setups, as well as for the orbifold Gromov-Witten theory of ${\bf P}^1$ (for a relevant recent progress and the quantum spectral curve in this case see \cite{Chen}).  Our findings, in particular the quantum spectral curves, should also be related to the topological recursion.  These topics will be discussed elsewhere.

The present paper is organized as follows. In Section \ref{S1} we remind the reader the free fermion description of the KP/Toda integrable hierarchies. Section \ref{seceq} is devoted to the description of the extended stationary Gromov-Witten invariants of ${\bf P}^1$. In Section \ref{S4} we generalize the results to the case of the Gromov-Witten invariants of ${\bf P}^1$ relative to one point. 

\section{Tau-functions and free fermions}\label{S1}

In this section we give a brief description of the modified 
Kadomtsev-Petviashvili (MKP) hierarchy in terms of free fermions (or, equivalently, infinite wedge space). For more details see \cite{JMbook,AZ,A1,Sa,Segal} and references therein. 
\begin{remark}
Let us stress that our conventions do not completely coincide with those of \cite{OP}. For example, we use the fermions with integer indices (compared to semi-integer indices in \cite{OP}), and denote the components of the bosonic current by $J_k$ (compared to $\alpha_k$ in \cite{OP}).
\end{remark}

Let us introduce the free fermions  $\psi_n , \psistar_{n}$, $n\in \mathbb{Z}$, which satisfy the canonical anticommutation relations
\be\label{anti}
[\psi_n , \psi_m ]_+ = [\psistar_n, \psistar_m]_+=0, \quad
[\psi_n , \psistar_m]_+=\delta_{mn}.
\ee
They generate
an infinite dimensional Clifford algebra. We use their
generating series
\be\label{ferm0}
\psi (z)=\sum_{k\in \z}\psi_k z^k, \quad \quad
\psistar (z)=\sum_{k\in \z}\psistar_k z^{-k}.
\ee

Next, we introduce a vacuum state $\left |0\rbr$, which is
a ``Dirac sea'' where all negative mode states are empty
and all positive ones are occupied:
\be
\psi_n \rvac =0, \quad n< 0; \quad \quad \quad
\psistar_n \rvac =0, \quad n\geq 0.
\ee
(For brevity, we call indices $n\geq 0$ {\it positive}.)
Similarly, the dual vacuum state has the properties
\be
\lvac \psistar_n  =0, \quad n< 0; \quad \quad \quad
\lvac \psi_n  =0, \quad n\geq 0.
\ee
With respect to the vacuum $\rvac$, the operators $\psi_n$ with
$n<0$ and $\psistar_n$ with $n\geq 0$ are annihilation operators
while, the operators $\psistar_n$ with $n<0$ and
$\psi_n$ with $n\geq 0$ are creation operators. The normal ordering $\normord (\ldots )\normord $ with respect
to the Dirac vacuum $\rvac$ is defined as 
follows: all annihilation operators
are moved to the right and all creation operators are moved to
the left, taking into account that the factor $(-1)$ appears 
each time two neighboring 
fermionic operators exchange their positions. 

We also introduce ``shifted'' Dirac vacua $\rvacn$ and $\lvacn$
defined as   
\begin{align}\label{vacdefr}
\rvacn = \left \{
\begin{array}{l}
\psi_{n-1}\ldots \psi_1 \psi_0 \rvac , \,\,\,\,\, n> 0,
\\ \\
\psistar_n \ldots \psistar_{-2}\psistar_{-1}\rvac , \,\,\,\,\, n<0,
\end{array} \right.
\end{align}
\begin{align}\label{vacdefl}
\lvacn = \left \{
\begin{array}{l}
\lvac \psistar_{0}\psistar_{1}\ldots \psistar_{n-1} , \,\,\,\,\, n> 0,
\\ \\
\lvac \psi_{-1}\psi_{-2}\ldots \psi_{n} , \,\,\,\,\, n<0.
\end{array} \right.
\end{align}
For them we have 
\begin{equation}
\begin{split}
 \psi_m \rvacn &=0, \quad m < n; 
\qquad 
\psistar_m \rvacn =0, \quad m \ge n, \\
\lvacn  \psi_{m}&=0 , \quad m \ge n; 
\qquad 
\lvacn  \psistar_{m}=0 , \quad m < n.
\end{split}
\end{equation}

Normally ordered bilinear combinations $X_B=\sum_{mn} B_{mn}\normord\psistar_m \psi_n\normord$
of the fermions, with certain conditions
on the matrix $B = (B_{mn})$, generate an
infinite-dimensional Lie algebra $\gl(\infty)$. Exponentiating these expressions, one obtains
an infinite dimensional group (a version
of $\GL(\infty )$) with the group elements 
\begin{equation}\label{gl}
G=\exp \Bigl (\sum_{i, k \in {\z }}B_{ik}\normord\psistar_i \psi_k\normord\Bigr ).
\end{equation}
For any group element $G$ a tau-function of the MKP hierarchy is given by a vacuum expectation value
\begin{equation}\label{tau}
\tau_n ({\bf t})=\lvacn e^{J_+ ({\bf t})}G\rvacn.
\end{equation}
It depends on
the variables ${\bf t}=\{t_1, t_2, \ldots \}$, usually called times, through the
 linear combination $J_+({\bf t})=\sum_{k>0}t_k J_k$ of the operators
\beq\label{Jk}
J_k =\sum_{j\in \z}\normord \psi_j \psistar_{j+k}\normord
=\mbox{res}_z \Bigl ( z^{-1} \normord \psi (z)
z^{k}\psistar (z)\normord \Bigr ).
\eeq
Here $\res_{z=0} \, z^k:=\delta_{k,-1}$. These operators are the Fourier modes of the {\em current operator}
$J(z)= z^{-1}\normord \psi (z)\psistar (z)\normord $ and span the Heisenberg algebra
\beq\label{Heis}
\left[J_k, J_l\right]= k \delta_{k+l,0}.
\eeq
Operators $J_k$ with positive and negative $k$ act on the vacuum as
\beq
J_k\rvac=\lvac J_{-k}=0,\quad \quad \quad k\geq0.
\eeq
For a fixed $n\in {\mathbb Z}$ tau-function (\ref{tau}) is a solution of KP hierarchy.  We assume that $\tau_n({\bf 0})=1$.  

Let $\mathbb{Z}_+$ be the set of all nonnegative integers. 
The MKP hierarchy relates $\tau_m$ to $\tau_{n}$ for any $m-n \in \mathbb{Z}_+$.
It can be described by the bilinear Hirota identity 
\begin{equation}\label{bi1}
\oint_{{\infty}} z^{m-n} e^{\sum_{k=1}^\infty (t_k-t_k')z^k}
\,\tau_{m} ({\bf t}-[z^{-1}])\,\tau_{n} ({\bf t'}+[z^{-1}])dz =0,
\end{equation}
where we use the standard short-hand notations
\be
{\bf t}\pm [z^{-1}]:= \bigl \{ t_1\pm   
z^{-1}, t_2\pm \frac{1}{2}z^{-2}, 
t_3 \pm \frac{1}{3}z^{-3}, \ldots \bigr \}.
\ee
More generally, for any group element $G$, the vacuum expectation value
\beq\label{stancor}
\tau_n({\bf t},{\bf s})=\lvacn e^{J_{+}({\bf t})} G e^{J_{-}({\bf s})}\rvacn ,
\eeq
where $J_-({\bf s})=\sum_{k>0} s_k J_{-k}$,
is a tau-function of the 2D Toda lattice hierarchy \cite{Ueno}.

Let $\Lambda:=\diag(\lambda_1,\dots,\lambda_N)$ be a diagonal matrix. In the Miwa parametrization 
\be
\left.f([\Lambda^{-1}]):=f({\bf t})\right|_{t_k=\frac{1}{k}\Tr \Lambda^{-k}}
\ee
the tau-function of the KP hierarchy can be represented as a ratio of two determinants
\be\label{taudet}
\tau_n([\Lambda^{-1}]) = \frac{\det_{k,l=1}^N \Phi_k^{(n)} (\lambda_l)}{\Delta(\lambda)},
\ee
where
\be
\Delta(\lambda)=\prod_{i<j}(\lambda_j-\lambda_i)
\ee
is the Vandermonde determinant. In (\ref{taudet}) the set of the basis vectors $\Phi_k^{(n)}$
\be
{\mathcal W}_n= \sppan_{\cc} \{\Phi_1^{(n)},\Phi_2^{(n)},\Phi_3^{(n)},\dots\} \in \rm{Gr}^{(0)}_+
\ee
defines a point of the Sato Grassmannian labeled by $n \in {\mathbb Z}$. We assume that the basis vectors are normalised:
\be\label{bvnorm}
\Phi_k^{(n)}(z)=z^{k-1}\left(1+O(z^{-1})\right).
\ee
\begin{remark}
The version of the Sato Grassmannian we are working with is a ``dual" one in the standard notations.  For simplicity, we omit the world "dual" below.
\end{remark}

Let us consider an algebra $w_{1+\infty}$ of the differential operators on the circle 
\be
w_{1+\infty}:= \sppan_{\cc} \left\{z^k D^m \middle| k\in \mathbb{Z}, m \in \mathbb{Z}_+\right\},
\ee
where 
\be
D:=z \p_z,
\ee
and $\p_z:=\frac{\p}{\p z}$. We also introduce
\be
w^\pm :=\sppan_{\cc} \left\{z^{\pm k} D^m \middle| k=1,2,3,\dots,m \in \mathbb{Z}_+\right\}.
\ee
and 
\be
w^{0}:= \sppan_{\cc}\left\{ D^m \middle| m \in  \mathbb{Z}_+  \right\}
\ee
so that
\be
w_{1+\infty}=w^-\oplus w^0 \oplus w^+.
\ee

For an operator $a\in w$ let
\be\label{btof}
W_a:=\res_z \left(z^{-1} \normord\psi(z) a \psi^*(z)  \normord \right),
\ee
be a corresponding fermionic $\gl(\infty)$ operator. For example, $J_k=W_{z^k}$.  Then
\begin{equation}
\begin{split}\label{Vacom}
\left[W_a,\psi(z)\right]&=a^* \cdot \psi(z),\\
\left[W_a,\psistar(z)\right]&=- a  \cdot \psistar(z),
\end{split}
\end{equation}
where for any monomial $P=z^k D^l$  the adjoint operator is $P^*=\left(-D\right)^l z^k$. By $\cdot$ we denote the action of the operator on the function to distinguish it from the product of operators.
This relation can be naturally lifted to the relation between the groups $e^{w_{1+\infty}}$ and
$\GL(\infty )$. Let us denote by ${\mathcal W}^G_n$ the point of the Sato Grassmannian, associated with the tau-function (\ref{tau}).  Then, if $a\in w^-$, from the construction in \cite{A1} it follows that 
\begin{lemma} \label{lemmaonaact}
For any $G$ the points of the Sato Grassmannians $ {\mathcal W}^G_n$ and ${\mathcal W}^{e^{W_a}G}_n$ for the group elements $G$ and $e^{W_a}G$, respectively,  are related by
\be
 {\mathcal W}^{e^{W_a}G}_n  =z^{n}\, e^{z^{-1} a z}\, z^{-n} \, {\mathcal W}^G_n.
\ee
\end{lemma}

\section{Gromov-Witten invariants of $\bf{P}^1$}\label{seceq}

Here we briefly describe the generating function of the equivariant Gromov-Witten invariants and its nonequivariant limit, for more details see \cite{OP,OP2} and references therein. Let $0$ and $\infty$ be the fixed points of the torus action on ${\bf P}^1$.  The torus ${\mathbb C}^*$ also acts canonically  on $\overline{\mathcal M}_{g,n} ({\bf P}^1,d)$ by translating the maps. The equivariant Poincar{\'e} duals of these points, ${ \boldsymbol{0}},{\boldsymbol{\infty}}\in H_{{\cc}^*}^2({\bf P}^1,{\mathbb Q})$ form a canonical basis in the equivariant cohomology of ${\bf P}^1$. We denote
\be\label{GWinv}
\left<\prod_{i=1}^m \tau_{k_i} ({0}) \prod_{i=m+1}^{n} \tau_{k_i} ({\infty}) \right>_{g,d}^{{\cc}^*}:=\int_{[\overline{\mathcal M}_{g,n} ({\bf P}^1,d)]^\vir} \prod_{i=1}^m  \psi_i^{k_i} \ev_i^*(\boldsymbol{0})\prod_{i=m+1}^{n}  \psi_i^{k_i} \ev_i^*(\boldsymbol{\infty}).
\ee
Equivariant Gromov-Witten generating function depends on two infinite sets of variables $x_i$ and $x_i^\star$ and generates all equivariant Gromov-Witten invariants (\ref{GWinv}),
\be\label{eqGW}
Z({\bf x},{\bf x^\star})=\exp\left(\sum_{g=0}^\infty\sum_{d=0}^\infty \hbar^{2g-2} q^d\left<\exp
\left(\sum_{k=0}^\infty x_k \tau_k(\boldsymbol{0})+x_k^\star\tau_k(\boldsymbol{\infty}) \right)\right>_{g,d}^{{\cc}^*}\right).
\ee
It is a formal series, $Z({\bf x},{\bf x^\star}) \in {\mathbb Q}[\epsilon][[{\bf x},{\bf x^\star},q,\hbar^2,\hbar^{-2}]]$, where we denote the equivariant parameter by $\epsilon$.
This generating function is a tau-function of the 2D Toda lattice hierarchy \cite{OP}, see Section \ref{Eqff} below.

The nonequivariant generating function (\ref{neqpf}) can be obtained from (\ref{eqGW}) by a nonequivariant limit \cite{OP}. Namely, we consider the linear change of variables
\be
x_k=\frac{1}{\epsilon}t_k^1, \,\,\,\,\,\,\,\,\,\,\,\,\,\,\,\,\,\,
x_k^\star=t_k^\omega-\frac{1}{\epsilon}t_k^1.
\ee
Then
\be\label{noneqpf}
\tau({\bf t^\omega},{\bf t^1}):=Z\left(\frac{1}{\epsilon}{\bf t^1},{ \bf t^\omega}-\frac{1}{\epsilon}{\bf t^1}\right)\Big|_{\epsilon=0}.
\ee
From the geometric construction it follows that the limit is smooth and there is no singularity at $\epsilon=0$.


\subsection{Free fermion description of equivariant $\bf{P}^1$ theory}\label{Eqff}

Let us remind a reader the description of the generating function of the equivariant Gromov-Witten invariants of $\bf{P}^1$ in terms of free fermions, constructed by Okounkov and Pandharipande \cite{OP}.  Let 
\be
\mathcal{E}_r(z):=\sum_{k\in \z}e^{z\left(k-\frac{r}{2}+\frac{1}{2}\right)} \normord\psi_{k-r}\psi_k^* \normord+\frac{\delta_{r,0}}{\varsigma(z)},
\ee
where
\be
\varsigma(z)=e^{z/2}-e^{-z/2}.
\ee
The commutation relations between the bosonic current components (\ref{Jk}) and $\mathcal{E}_r(z)$ follow from (\ref{anti}):
\be\label{commpe}
\left[J_k,{\mathcal E}_l(z)\right]=\varsigma(kz) {\mathcal E}_{k+l}(z).
\ee

The components of the operator $\mathcal{E}_0$,
\be\label{Paspsi}
{\mathcal P}_k:=k! \left[z^k\right] \mathcal{E}_0(z),
\ee
constitute a commutative subalgebra in $\gl(\infty)$,
\be
\left[{\mathcal P}_k, {\mathcal P}_m\right]=0,\,\,\,\,\ k,m\in \ZZ_+.
\ee
These operators are diagonal
\be
{\mathcal P}_k=\sum_{m\in \z} (m+1/2)^k\normord\psi_{m}\psi_m^* \normord +\frac{B_{k+1}(1/2)}{k+1},
\ee
where $B_k(n)$ are the Bernoulli polynomials, defined by
\be\label{berndef}
\frac{ze^{nz}}{e^z-1}=\sum_{k=0}^\infty B_k(n) \frac{z^k}{k!}.
\ee
The charged vacuum $\rvacn$ is the eigenstate of the operator $\mathcal{E}_0(z)$
\be\label{scv}
\mathcal{E}_0(z) \rvacn =\frac{e^{nz}}{\varsigma(z)} \rvacn,
\ee
hence
\be\label{Pactn}
{\mathcal P}_k \rvacn = \frac{B_{k+1}(n+1/2)}{k+1} \rvacn. 
\ee
The commutation relations between the operators ${\mathcal P}_k$ and $J_k$ follow from (\ref{commpe}), in particular 
\be\label{JPcom}
\left[J_k, {\mathcal P}_1\right]= k J_k.
\ee

Let us also consider
\be
\mathcal{A}(v,w):=\left(\frac{\varsigma(w)}{w}\right)^v\sum_{k\in \z} \frac{\varsigma(w)^k}{(v+1)_k}\mathcal{E}_k(w),
\ee
where
\be
(1+z)_k:=\frac{\Gamma(z+k+1)}{\Gamma(z+1)}
\ee
is the Pochhammer symbol.
Following \cite{OP} we introduce the vacuum expectation value
\be\label{corfn}
Z_n({\bf x},{\bf x^\star}):=\lvacn e^{\sum_{k=0}^\infty x_k {\mathsf A}_k}\, G_0\, e^{\sum_{k=0}^\infty x_k^\star {\mathsf A}_k^\star}\rvacn,
\ee 
where
\be
G_0:=e^{J_1}\left(\frac{q}{\hbar^2}\right)^{{\mathcal P}_1+\frac{1}{24}} e^{J_{-1}}
\ee
is a $GL(\infty)$ group element. Operators ${\mathsf A}_k$ and ${\mathsf A}_k^\star$ from $\gl(\infty)$ can be described in terms of their generating functions
\begin{equation}
\begin{split}
{\mathsf A}(z)&=\frac{1}{\hbar}\mathcal{A}(\epsilon z, \hbar z),\\
{\mathsf A}^\star(z)&=\frac{1}{\hbar}\mathcal{A}(-\epsilon z, \hbar z)^*,
\end{split}
\end{equation}
namely
\be
{\mathsf A}_k=\left[z^{k+1}\right]{\mathsf A}(z),\,\,\,\, {\mathsf A}_k^\star=\left[z^{k+1}\right]{\mathsf A}^\star(z),\,\,\,\,k\in\ZZ.
\ee
The following theorem describes the equivariant Gromov-Witten theory of ${\bf P}^1$ in the free fermion formalism \cite{OP}:
\begin{theorem*}[Okounkov, Pandharipande]
The equivariant generating function (\ref{eqGW}) is given by the vacuum expectation value (\ref{corfn}),
\be
Z({\bf x},{\bf x^\star})=Z_0({\bf x},{\bf x^\star}).
\ee
\end{theorem*}
In \cite{OP} it was shown that there are conjugation operators ${\mathsf W}(\epsilon)\in GL(\infty)$ and ${\mathsf W}^\star(\epsilon)\in GL(\infty)$ such that
\be\label{epsconj}
W(\epsilon)^{-1}\left( \sum_{k=0}^\infty x_k {\mathsf A}_k\right)W(\epsilon) = J_+({\bf t }),\\
W^\star(\epsilon)\left( \sum_{k=0}^\infty x^\star_k {\mathsf A}_k^\star\right)W^\star(\epsilon)^{-1} =J_-({\bf s }).
\ee
The operators $W(\epsilon)$ and $W^\star(\epsilon)$ belong to the upper and lower triangular subgroups of $\GL(\infty)$ respectively, so $\lvac W(\epsilon)=\lvac$, $W^\star(\epsilon) \rvac = \rvac$. The variables ${\bf t}$ and ${\bf s}$ are related to the variables ${\bf x}$ and ${\bf x^\star}$ by a linear transformation, conjectured by Getzler \cite{Getz1},
\be\label{timechange}
t_n=\hbar^{n-1}\, \res_{z=0} \sum_{k=0}^\infty x_k \frac{z^{n-k-2}}{(1+\epsilon z)\dots(n+\epsilon z)},\\
s_n=\hbar^{n-1}\, \res_{z=0} \sum_{k=0}^\infty x_k^\star \frac{z^{n-k-2}}{(1-\epsilon z)\dots(n-\epsilon z)}.
\ee
Therefore, the correlation function (\ref{corfn}) being a function of the variables $\bf{t}$ and $\bf{s}$ has a canonical 2D Toda fermionic form (\ref{stancor}) with
\be
G=W(\epsilon)^{-1} G_0 W^\star(\epsilon)^{-1}.
\ee
Moreover, in \cite{OP} it is proven that dependence on the variable $n$ is described by
\be\label{ndep}
Z_n({\bf x},{\bf x^\star})=\left(\frac{q}{\hbar^2}\right)^{\frac{n^2}{2}} e^{\frac{\hbar n}{\epsilon} \left(\frac{\p}{\p x_0}-\frac{\p}{\p x_0^\star}\right)}Z({\bf x},{\bf x^\star}),
\ee
which defines the  left hand side of this equation not only for $n \in {\mathbb Z}$, but  for arbitrary value of the continuous parameter $n$.


\subsection{Extended stationary sector of Gromov-Witten theory of ${\bf P}^1$} 

The nonequivariant limit of the vacuum expectation value (\ref{corfn}) is highly non-trivial. However, to describe the extended stationary sector (\ref{tau11}) it is enough to take a naive limit, or just
to put $\epsilon=0$: 
\be\label{Azero}
{\mathsf A}(z)\Big|_{\epsilon=0}=\frac{1}{\hbar}\sum_{k\geq 0} \frac{\varsigma(\hbar z)^k}{k!} {\mathcal E}_k(\hbar z),\\
{\mathsf A}^\star(z)\Big|_{\epsilon=0}=\frac{1}{\hbar} \sum_{k\geq 0} \frac{\varsigma(\hbar z)^k}{k!} {\mathcal E}_{-k}(\hbar z), 
\ee
or, as follows from (\ref{commpe}), 
\begin{equation}\label{Azero1}
\begin{split}
{\mathsf A}(z)\Big|_{\epsilon=0}&=\frac{1}{\hbar}e^{J_1} {\mathcal E}_0(\hbar z)e^{-J_1},\\
{\mathsf A}^\star(z)\Big|_{\epsilon=0}&=\frac{1}{\hbar}e^{-J_{-1}} {\mathcal E}_0(\hbar z)e^{J_{-1}}.
\end{split}
\end{equation}
Let 
\be\label{Aconjj}
W:=W(0), \,\,\,\,\,\,\, W^\star:=W^\star(0).
\ee
Comparing (\ref{Azero1}) with (\ref{epsconj}) one concludes that
\be\label{Wconf}
W^{-1}\,e^{J_1}\,{\mathcal P}_k\,e^{-J_1}\,W =J_k,\\
W^\star\,e^{-J_{-1}}\,{\mathcal P}_k\,e^{J_{-1}}\,W^{\star-1}=J_{-k}.
\ee
Substituting the components of (\ref{Azero1}) into (\ref{corfn}), one gets
\begin{equation}
\begin{split}
Z_n({\bf x},{\bf x}^\star)\Big |_{\epsilon=0}&=\lvacn e^{J_1} e^{\sum_{k=1}^\infty x_{k-1} \frac{\hbar^{k-1} {\mathcal P}_{k}}{k!}}e^{-J_1}G_0 e^{-J_{-1}} e^{\sum_{k=1}^\infty x^\star_{k-1} \frac{\hbar^{k-1} {\mathcal P}_{k}}{k!}} e^{J_{-1}} \rvacn\\
&=\lvacn e^{J_1}e^{\sum_{k=1}^\infty (x_{k-1}+x^\star_{k-1}) \frac{\hbar^{k-1} {\mathcal P}_{k}}{k!}} \left(\frac{q}{\hbar^2}\right)^{{\mathcal P}_1+\frac{1}{24}}    e^{J_{-1}} \rvacn\\
&=\left(\frac{q}{\hbar^2}\right)^\frac{n^2}{2}\lvacn e^{J_1}e^{\sum_{k=1}^\infty (x_{k-1}+x^\star_{k-1}) \frac{\hbar^{k-1} {\mathcal P}_{k}}{k!}}e^{\frac{{q}J_{-1}}{\hbar^2}}\rvacn.
\end{split}
\end{equation}
In the last equality we use (\ref{Pactn}) and (\ref{JPcom}). Note, that $Z_n$ depends only on  the sum of the variables $x_k$ and $x_k^\star$.

Let us introduce
\be
\tau_n({\bf t}):=e^{-\tilde{q}}\,\tilde{q}^{-\frac{n^2}{2}} \,Z_n({\bf x},{\bf 0})\Big|_{\epsilon=0; \, x_k=\frac{(k+1)! t_{k+1}}{\hbar^k}},
\ee
where $t_k$ and $x_k$ are related by the nonequivariant limit of Getzler's change of variables (\ref{timechange}) and $\tilde{q}=q/\hbar^2$. Then
\be\label{taunnn}
\tau_n({\bf t})= e^{-\tilde{q}} \lvacn e^{J_1}   e^{\sum_{k=1}^\infty t_k  {\mathcal P}_{k} }e^{\tilde{q} J_{-1}} \rvacn.
\ee
It is related to the noneqivariant tau-function (\ref{tau11}) by a change of variables
\be\label{COV}
\left.\tau({\bf t^{\omega}},t_0^1)=e^{\tilde{q}} \tau_{t_0^1/\hbar} ({\bf t})\right|_{t_k=\frac{\hbar^{k-1}t^\omega_{k-1}}{k!}}.
\ee
Below we work with the tau-function $\tau_n({\bf t})$, and make the change of variables back to ${\bf t^{\omega}}$ and $t_0^1$ only in the final expression for the matrix integral.
This tau-function can also be expressed by vacuum expectation value (\ref{tau}) with standard bosonic current components,
\be\label{standardtn}
\tau_n({\bf t})= e^{-\tilde{q}}\,  \tilde{q}^{-\frac{n^2}{2}}\, \lvacn e^{J_+(\bf t)} W^{-1}G_0 \rvacn. 
\ee
Equations (\ref{Pactn})  and (\ref{JPcom}) yield
\be\label{G0act}
G_0\rvacn=\tilde{q}^{\frac{n^2}{2}}  \,    e^{\tilde{q}(J_{-1}+1)}\rvacn,
\ee
so that 
\be\label{fermm}
\tau_n({\bf{t}})= \lvacn e^{J_+(\bf t)} W^{-1} e^{\tilde{q}J_{-1}}  \rvacn.
\ee
Let us denote corresponding point of the Sato Grassmannian by ${\mathcal W}_{n}$.
\begin{remark}
Tau-function (\ref{taunnn}) depends on $\hbar$ and $q$ only through the combination $\tilde{q}$.
\end{remark}


\subsection{Tau-function at $q=0$}\label{q0limit}
As a warm-up example, let us find the  tau-function (\ref{taunnn}) at $q=0$.
\begin{remark}
At $q=0$ the equivariant generating function (\ref{eqGW}) remains non-trivial, but factorises \cite{OP}.
The generating function of the degree $0$ equivariant Gromov-Witten invariants of ${\bf P}^1$ coincides with the generating function of linear Hodge integrals. We will discuss this relation in the forthcomming publication.
\end{remark}
We have
\be\label{tauforq0}
\tau_n({\bf t})\Big|_{q=0}=  \lvacn e^{J_1}  e^{\sum_{k=1}^\infty t_k  {\mathcal P}_{k} } \rvacn=e^{\sum_{k=1}^\infty c_k(n) t_k},
\ee
where
\be\label{ck}
c_k(n)=\lvacn  {\mathcal P}_{k} \rvacn =\frac{B_{k+1}(n+1/2)}{k+1}.
\ee
Below we will need the value of $c_k$ for $k=1$
\be\label{ck1}
c_1(n)=\frac{B_2(n+1/2)}{2}=\frac{n^2}{2}-\frac{1}{24}.
\ee

Let us introduce
\be
F_n(z):=\sum_{k=1}^\infty \frac{c_k(n)}{kz^k}=\sum_{k=1}^\infty \frac{B_{k+1}(n+1/2)}{(k+1)kz^k}.
\ee
Then
\be
F_n(z)=\sum_{k=1}^\infty (-1)^{k+1}\frac{B_{k+1}(1/2-n)}{(k+1)kz^k},
\ee
where we use a symmetry of the Bernoulli polynomials, obvious from the definition (\ref{berndef}).
This series appears in Stirling's expansion of the gamma function,
\be\label{Gammaas}
\Gamma(z+1/2-n)=\sqrt{2\pi}z^{z-n}e^{-z}e^{F_n(z)},
\ee
valid for large values of $|z|$ with $|\arg(z)|<\pi$. Here $n$ is an arbitrary finite complex number. 

\begin{remark} 
Let us stress that (\ref{Gammaas}) gives the asymptotic expansion, which acquires the nonperturbative corrections of the form $e^{\pm2\pi i z}$, see, e.g., \cite{Gammacor}.
Below we always neglect the nonperturbative corrections and identify the gamma-function with its asymptotic expansion (\ref{Gammaas}).
\end{remark}
Principal specialization of the tau-function is given by
\be
\tau_n([z^{-1}])\Big|_{q=0}=\frac{\Gamma(z+1/2-n)}{\sqrt{2\pi}z^{z-n}e^{-z}}.
\ee
Hence, the point of the Sato Grassmannian ${\mathcal W}_{n}\Big|_{q=0}$ for the tau-function (\ref{tauforq0}) is given by the basis vectors
\be
\tilde{\Phi}_k^{(n)}(z)\Big|_{q=0}=\frac{\Gamma(z+1/2-n)}{\sqrt{2\pi}z^{z-n}e^{-z}}z^{k-1}.
\ee
It is more convenient to consider another basis in the same space, 
\be\label{qzerobv1}
{\Phi}_k^{(n)}(z)\Big|_{q=0}=\frac{\Gamma(z+k-1/2-n)}{\sqrt{2\pi}z^{z-n}e^{-z}}.
\ee


\subsection{Conjugation operators at $\epsilon=0$}

In this section we describe the conjugation operators $W$ and $W^{\star}$. Let us focus on $W$, while the description of the  adjoint operator  $W^{\star}$ can be obtained similarly.
For this purpose we consider two different descriptions of the operator ${\mathsf A}_0$  at $\epsilon=0$. 
On the one hand, from (\ref{epsconj}) it follows that
\be
{\mathsf A}_0\Big|_{\epsilon=0}  = W J_1 W^{-1}.
\ee
On the other hand, from (\ref{Azero1}) we have
\be
{\mathsf A}_0\Big|_{\epsilon=0}=e^{J_1}{\mathcal P}_1 e^{-J_1}=J_1+ {\mathcal P}_1.
\ee
In terms of fermions (\ref{btof}) this operator is given by
\be\label{A0ferm}
{\mathsf A}_0\Big|_{\epsilon=0}=\res_z \left(z^{-1} \normord \psi(z){\mathsf a}_0 \psi^*(z) \normord\right)-\frac{1}{24},
\ee
where 
\be
{\mathsf a}_0=z+\frac{1}{2}-D.
\ee

Let us construct an operator ${\mathtt W}\in \exp(w^-)$ acting on the Sato Grassmannian, which is a counterpart of the dressing operator $W$ in the sense of Lemma \ref{lemmaonaact}. For this purpose, we find an operator $\tilde{\mathtt W} \in \exp(w^-)$ such that
\be\label{wrv}
\tilde{\mathtt W}\, z\, \tilde{\mathtt W}^{-1}= z^{-1} {\mathsf a}_0 z =  z-\frac{1}{2}- D.
\ee

\begin{remark}
As it was indicated in \cite{OP}, this equation does not completely specify the solution. However, it is easy to show that solution $\tilde{\mathtt W}$ is related to ${\mathtt W}$ by a right multiplication with $e^{{\mathtt f}(z)}$,
\be
{\mathtt W}=\tilde{\mathtt W}\, e^{{\mathtt f}(z)},
\ee
where
\be\label{corseries}
{\mathtt f}(z)=- \sum_{k=1}^\infty \alpha_k z^{-k}
\ee
is an operator of multiplication by a formal series with $\alpha_k \in \mathbb{C}$. In the central extended version this leads to a relation
\be
W=\tilde{W} e^{- \sum_{k=1}^\infty \alpha_k J_{-k}}
\ee
and
\begin{equation}
\begin{split}
\tau_n({\bf t})&=  \lvacn e^{J_+(\bf t)} e^{\sum_{k=1}^\infty \alpha_k J_{-k}} \tilde{W}^{-1} e^{\tilde{q} J_{-1}} \rvacn\\
&= e^{ \sum_{k=1}^\infty k\alpha_k t_k}  \lvacn e^{J_+(\bf t)} \tilde{W}^{-1} e^{\tilde{q} J_{-1}} \rvacn.
\end{split}
\end{equation}
Thus, different choices of $\tilde{W}$ correspond to the multiplication of the tau-function by the factor $\exp(\sum k \alpha_k t_k)$, and, in particular, are responsible for the constant contributions in (\ref{A0ferm}). Since the coefficients of the operator $\mathcal{A}$ do not depend on $q$, we can find $\alpha_k$'s  from the consideration of the tau-function at $q=0$. Series ${\mathtt f}(z)$ for a particular choice of $\tilde{\mathtt W}$ will be determined in the next section. 
\end{remark}

Let us put
\be
\tilde{\mathtt W}=e^{\mathtt y}e^{-\frac{1}{2}\p_z},
\ee
where ${\mathtt y} \in w_-$. Then from (\ref{wrv}) it follows that the operator ${\mathtt y}$ satisfies
\be\label{ydef}
e^{\mathtt y}\, z\, e^{-\mathtt y}=z-z\p_z.
\ee
It is easy to see that the ansatz 
\be
{\mathtt y}=z\, g\left( \p_z\right),
\ee
for some $g(z)\in z^2{\mathbb Q}[[z]]$, provides a solution. The coefficients of the series $g$ can be obtained recursively. 
The adjoint operator $W^{\star}$ can be constructed similarly with 
\be
{\mathtt y}^\star=z^{-1} g\left(z^{2}\p_z\right).
\ee
Hence, we proved
\begin{lemma}\label{l_WW}
\begin{equation}
\begin{split}
{\mathtt W}&=e^{\mathtt y} e^{-\frac{1}{2}\p_z} e^{{\mathtt f}(z)},\\
{\mathtt W}^\star&=e^{{\mathtt f}(z^{-1})} e^{-\frac{z^2}{2}\p_z} e^{{\mathtt y}^\star}.
\end{split}
\end{equation}
\end{lemma}
From Lemmas \ref{lemmaonaact} and \ref{l_WW} we have
\begin{corollary}
\begin{equation}
\begin{split}
W&=e^{W_{z g(\p_z - z^{-1})}} e^{W_{-\frac{1}{2}\left(\p_z-z^{-1}\right)}} e^{- \sum_{k=1}^\infty \alpha_k J_{-k}},\\
W^\star&= e^{- \sum_{k=1}^\infty \alpha_k J_{k}} e^{W_{-\frac{z^2}{2}(\p_z-z^{-1})}} e^{W_{z^{-1}g(z^2(\p_z-z^{-1}))}}.
\end{split}
\end{equation}
\end{corollary}
To prove the corollary one has only to check that the central extension does not change the normalization of the tau-function, that is $\tau_n({\bf 0})=1$.

\begin{remark}
Coefficients of the series $g$ coincide (up to a prefactor $(-1)^k/k!$) with the coefficients $c_k$, introduced in \cite{SZ} for the description of the linear change of variables connecting the generating functions of the Hurwitz numbers and Hodge integrals,
\be
g(z)=-\frac{1}{2}\,{z}^{2}-\frac{1}{12}\,{z}^{3}-\frac{1}{48}\,{z}^{4}-{\frac {1}{180}}\,{z}^{5}-{
\frac {11}{8640}}\,{z}^{6}-{\frac {1}{6720}}\,{z}^{7}+{\frac {11}{
241920}}\,{z}^{8}\\
+{\frac {29}{1451520}}\,{z}^{9}-{\frac {493}{43545600
}}\,{z}^{10}-{\frac {2711}{239500800}}\,{z}^{11}+O(z^{12}).
\ee
\end{remark}

Let us describe the operator ${\mathtt y}$ in more detail.

\begin{lemma}\label{Commu}
\begin{align}\label{commm1}
e^{-\mathtt y}\p_z e^{\mathtt y}&= 1-e^{-\p_z},\\
e^{-\mathtt y}ze^{\mathtt y}&=z e^{\p_z}.
\label{zconj}
\end{align}
\end{lemma}
\begin{proof}
Let us denote $e^{\mathtt y}\p_ze^{- \mathtt y}\in {\mathbf C}[[\p_z]]$ by $p(\p_z)$. Since
\be
\left[e^{\mathtt y}\p_ze^{-\mathtt y},e^{\mathtt y} z e^{-\mathtt y}\right]=1,
\ee
we have
\be
\left[p(\p_z),z-z\p_z\right]=p'(\p_z)(1-\p_z)=1.
\ee
Hence
\be
p(z)=\int\frac{dz}{1-z} =-\log(1-z),
\ee
and $e^{\mathtt y}\p_ze^{-\mathtt y}=-\log(1-\p_z)$. Therefore,
\be\label{y2c}
e^{\mathtt y}e^{-\p_z}e^{- \mathtt y}=1-\p_z
\ee
and (\ref{commm1}) immediately follows. Now, from (\ref{ydef}) it follows that
\be
e^{-\mathtt y}z(1-\p_z) e^{\mathtt y} =z
\ee
and combining it with (\ref{y2c}) one obtains (\ref{zconj}).
\end{proof}

\begin{remark}
This Lemma gives an alternative description of the series $g$, namely
\be
e^{g(z)\p_z} \cdot z= 1-e^{-z}.
\ee
\end{remark}


\subsection{Basis vectors and Kac-Schwarz operators}
Let $\Phi_k^{(n)}(z)$ be the basis vectors for the KP tau-function (\ref{fermm}). Then, using Lemma \ref{lemmaonaact}, we get
\be\label{basisvec}
\Phi_k^{(n)}(z)=z^{n}\, {\mathtt W}^{-1}\cdot e^{\frac{\tilde{q}}{z}}
 z^{k-n-1} ,\,\,\,\,\,k=1,2,3,\dots.
\ee
The series (\ref{corseries}) is fixed by the comparison with (\ref{qzerobv1}), namely it should satisfy the equations
\be
z^n e^{-\mathtt f(z)}\, e^{\frac{1}{2}\frac{\p}{\p z}}e^{-\mathtt y}\cdot z^{k-n-1}=\frac{\Gamma(z+k-1/2-n)}{\sqrt{2\pi}z^{z-n}e^{-z}}.
\ee
Using (\ref{zconj}) for any $k\in \mathbb Z$ we get a rational function of $z$,
\begin{equation}\label{yaction}
\begin{split}
e^{-\mathtt y} \cdot z^{k}&= \left(e^{-\mathtt y} z e^{\mathtt y} \right)^{k-1} \cdot z\\
&=\left(z e^{\p_z}\right)^{k-1} \cdot z\\
&=\frac{\Gamma(z+k)}{\Gamma(z)}.
\end{split}
\end{equation}
Thus,
\be\label{fser}
e^{-{\mathtt f}(z)}=\frac{\Gamma(z+1/2)}{\sqrt{2\pi}z^{z}e^{-z}}.
\ee
We see that $\mathtt f$ does not depend on $k$ or $n$, as it should be, and $\alpha_k=\frac{B_{k+1}(1/2)}{(k+1)k}$.
Now we can completely describe an action of the operator ${\mathtt W}^{-1}$ on $\mathbb{C}((z^{-1}))$:
\begin{lemma}\label{lemmainttran}
Operator ${\mathtt W}^{-1}$ defines an integral transform on the space of formal Laurent series $\mathbb{C}((z^{-1}))$ given by
\be
{\mathtt W}^{-1} \cdot  g(z)=\frac{z^{-z}e^z}{\sqrt{2\pi}} \int_{\rr} e^{y(z+1/2)-e^y} g(e^y)dy.
\ee

\end{lemma}
\begin{proof}
From (\ref{yaction}) and (\ref{fser}) for any $k\in \mathbb Z$ we have
\begin{equation}
\begin{split}
{\mathtt W}^{-1}\cdot z^k  &=\frac{\Gamma(z+1/2)}{\sqrt{2\pi}z^{z}e^{-z}}\frac{ \Gamma(z+k+1/2)}{\Gamma(z+1/2)}\\
&=\frac{\Gamma(z+k+1/2)}{\sqrt{2\pi}z^{z}e^{-z}},
\end{split}
\end{equation}
where we consider the asymptotic expansion of the right hand side with the help of (\ref{Gammaas}). 

Now we can apply the standard integral representation of the gamma function, valid for $\Re(z)>0$:
\be
\Gamma(z)=\int_{0}^\infty x^{z-1}e^{-x} dx,
\ee
which in terms of the variable $y=\log x$ can be represented as
\be
\Gamma(z)=\int_{\rr} e^{yz-e^y} dy.
\ee
Therefore,
\be
{\mathtt W}^{-1}\cdot z^k  =\frac{z^{-z}e^z}{\sqrt{2\pi}} \int_{\rr} e^{y(z+k+1/2)-e^y} dy.
\ee
where we take the asymptotic expansion of the integral at large positive $z$. This completes the proof.
\end{proof}

In particular, for the basis vectors (\ref{basisvec}) we have
\begin{corollary}
\be\label{intrep}
\Phi_k^{(n)}(z)= \frac{z^{n-z}e^z}{\sqrt{2\pi}} \int_{\rr} e^{y(z+k-n-1/2)-e^y+\tilde{q}e^{-y}} dy.
\ee
\end{corollary}
In the right hand side we take a series expansion of the integrand in $\tilde{q}$ and then take the integrals term by term using asymptotic expansion (\ref{Gammaas}):
\be
\Phi_k^{(n)}(z)=\frac{z^{n-z}e^{z}}{\sqrt{2\pi}}\sum_{l=0}^\infty\frac{\tilde{q}^l}{l!}\Gamma(z+k-n-l-1/2)  \in z^{k-1} +z^{k-2}\mathbb{C}[[\tilde{q},z^{-1}]].
\ee
We also have an equivalent expression
\be\label{logint}
\Phi_k^{(n)}(z)= \frac{z^{n-z}e^z}{\sqrt{2\pi}} \int_{0}^\infty x^{z+k-n-3/2}e^{-x+\frac{\tilde{q}}{x}} dx.
\ee

Let us find the Kac-Schwarz operators associated with ${\mathcal W}_n$. It is obvious that
\be
\Phi_{k+1}^{(n)}(z)= \frac{z^{n-z}e^z}{\sqrt{2\pi}} e^{\p_z} \int_{\rr} e^{y(z+k-n-1/2)-e^y+\tilde{q}e^{-y}} dy.
\ee
Thus, the operator
\begin{equation}
\begin{split}
b:&= z^{n-z}\,e^z e^{\p_z} z^{z-n}\,e^{-z}\\
&=\left(\frac{z}{z+1}\right)^{n-z} (z+1)e^{\p_z-1}
\end{split}
\end{equation}
stabilizes the point of the Sato Grassmannian,
\be\label{recurs}
b\cdot \Phi_k^{(n)}(z)= \Phi_{k+1}^{(n)}(z).
\ee
Hence, $b$ is the recursion operator.
For the large values of $|z|$ we have
\be\label{bexp}
b e^{-\p_z}=z+\frac{1}{2}-n+\left( \frac{1}{2}\,{n}^{2}-\frac{1}{24}\right){z}^{-1}+O \left( {z}^{-2} \right). 
\ee

Operator $b$ is the Kac-Schwarz operator, obtained by conjugation of $z$
\be
b=z^{n}\,{\mathtt W}^{-1} e^{\frac{\tilde{q}}{z}}z e^{-\frac{\tilde{q}}{z}}{\mathtt W}\, z^{-n}.
\ee
Another Kac-Schwarz operator can be obtained by conjugation of $\p_z+\frac{n}{z}$,
\be
a:=z^{n}\, {\mathtt W}^{-1} e^{\frac{\tilde{q}}{z}}\left(\p_z+\frac{n}{z}\right) e^{-\frac{\tilde{q}}{z}}{\mathtt W}\, z^{-n}.
\ee
Using Lemma \ref{Commu}, we get
\be
a=1+\tilde{q}b^{-2}+(n+1/2-z)b^{-1},
\ee
and
\be\label{aeq}
a \cdot \Phi_k^{(n)}(z)  =(k-1)\, \Phi_{k-1}^{(n)}(z).
\ee
An operator inverse to $b$, 
\be
b^{-1}=\left(\frac{z}{z-1}\right)^{n-z}\frac{1}{z-1}e^{1-\p_z},
\ee
is not a Kac-Schwarz operator.
\begin{proposition}\label{KSlemma}
Operators $a$ and $b$ are the Kac-Schwarz operators for ${\mathcal W}_n$,
\be
a\cdot {\mathcal W}_n \in {\mathcal W}_n\\
b\cdot {\mathcal W}_n \in {\mathcal W}_n.
\ee
They satisfy the commutation relation $\left[a,b\right]=1$ and completely specify a point ${\mathcal W}_n$ of the Sato Grassmannian.
\end{proposition}

\begin{proof}
It remains to prove that ${\mathcal W}_n$ is specified uniquely by $a$ and $b$. From the expansion (\ref{bexp}) it follows that
\be\label{aaction}
a \cdot z^{k} =k \,z^{k-1} \left(1+O(z^{-1})\right).
\ee
Hence, if $a$ is a KS operator for some tau-function, it should annihilate the first basis vector
\be
a \cdot \Phi_1(z)=0.
\ee
It is easy to see, that there is a normalized unique solution of this equation in the space of the Laurent series in $z^{-1}$. Moreover, this solution is a formal series in $\tilde q$ by construction.

Then, since
\be
b\cdot  z^k= z^{k+1}\left(1+O(z^{-1})\right)
\ee
all higher basis vectors can be obtained by application of $b$ to the first basis vector. Since operator $b$ does not depend on $\tilde{q}$, all higher vectors constructed by recursion are also formal series in $\tilde{q}$. 
\end{proof}

The equation
\be
\left(\frac{n}{z}-\frac{\p}{\p z}-\frac{\p}{\p n}\right)\Phi_k^{(n)} =0
\ee
easily follows from the integral representation (\ref{intrep}). It is the Sato Grassmannian version of the string equation. Namely, on the level of tau-function the operator  $-\frac{\p}{\p z}$ corresponds to a component of the Virasoro algebra $L_{-1}=\sum_{k=1}^\infty k  t_k \frac{\p}{\p t_{k-1}}$, and after the change of variables (\ref{COV}) we get the string equation
\be
\left(\frac{t_{0}^1t_0^{\omega}}{\hbar^2} +\sum_{k=1}^\infty t_k ^{\omega}\frac{\p}{\p t_{k-1}^\omega} -\frac{\p}{\p t_0^1} \right) \tau({\bf t}^\omega,t_0^1)=0.
\ee

\subsection{Quantum spectral curve and matrix model}

Let us modify the basis vectors (\ref{intrep}) by the non-stable contributions \cite{Shadrin}
\be\label{modff}
\Psi_k(z,n):=e^{\frac{1}{\hbar}\left(z\log(z)-z\right)-n\log(z)}\Phi_k^{(n)}\left(\frac{z}{\hbar}\right).
\ee
These modified functions can be represented as
\be\label{waveint}
\Psi_k(z,n):=\frac{\hbar^{1/2-k}}{\sqrt{2\pi}}\int_{\rr} e^{\frac{1}{\hbar} \left(y(z+\hbar(k-n-1/2))-e^y+qe^{-y}\right)}dy.
\ee
They depend on $z$ and $n$ only throught the combination $z-\hbar n$. 

Note that the basis vectors satisfy 
\be
a\,b\cdot \Phi_k^{(n)} = k \, \Phi_k^{(n)},
\ee
where
\be
a\,b= b+\tilde{q} b^{-1} +n+1/2-z
\ee
is a Kac-Schwarz operator. From this equation, or directly from the integral representation (\ref{waveint}), it follows that the modified functions (\ref{modff}) satisfy the equations
\be
\left(e^{\hbar \p_z}+q e^{-\hbar \p_z}-z+\hbar(n+\frac{1}{2} -k)\right)\Psi_k(z,n)=0.
\ee
For $k=1$ it coincides with the equation for the quantum spectral curve obtained for $n=1/2$ in \cite{Shadrin} (in the dual parametrisation) and for generic $n$ in 
 \cite{Norb}. This quantum spectral curve corresponds a to classical curve:
 \be
 e^{y}+qe^{-y}=x.
 \ee


\begin{proof}[{\bf Proof of Theorem \ref{T_1}}]
Let us derive the matrix integral representation of the tau-function (\ref{taunnn}) from the determinant formula (\ref{taudet}). Using the basis vectors (\ref{intrep}) one gets 
\be\label{ehigenv}
\tau_n\left(\left[\Lambda^{-1} \right]\right)=\frac{1}{(2\pi)^\frac{N}{2} \Delta(\lambda) {\mathcal P}} \int_{\rr^N} \prod_{k=1}^N d y_k
\, \Delta(e^y)e^{\Tr \left(\tilde{Y}(\Lambda-n+1/2)-e^{\tilde{Y}}+\tilde{q} e^{\tilde{Y}}\right)},
\ee
where
\be
{\mathcal P}:= e^{-\Tr \left((n-\Lambda)\log\Lambda+\Lambda\right) },
\ee
and $\tilde{Y}:=\diag(y_1,\dots,y_N)$.
One can use the Harish-Chandra-Itzykson-Zuber integral to transform the eigenvalue integral (\ref{ehigenv}) into the matrix integral. Namely, let us normalize the Haar measure on the unitary group $U(N)$ by $\int_{U(N)} \left[d U\right] =1$,
then for two diagonal matrices $A=\diag(a_1,\dots,a_N)$ and $B=\diag(b_1,\dots,b_N)$ we have an identity
\be
\int_{U(N)} \left[d U\right] e^{\Tr U A U^\dagger B}=\left(\prod_{k=1}^{N-1}k! \right)\frac{\det_{i,j=1}^N e^{a_ib_j}}{\Delta(a)\Delta(b)}.
\ee
Then
\be\label{tauT}
\tau_n\left(\left[\Lambda^{-1}\right]\right)= \frac{1}{(2\pi)^\frac{N}{2} {\mathcal P}\prod_{k=1}^Nk!}\int_{\rr^N} \prod_{k=1}^N d y_k \int_{U(N)} \left[d U\right] \Delta(e^y)\Delta(y)e^{\Tr \left(\tilde{Y}U^\dagger(\Lambda-n+1/2)U-e^{\tilde{Y}}+\tilde{q}e^{\tilde{Y}}\right)}.
\ee
Let $Y=U\tilde{Y}U^\dagger$ be a Hermitian matrix. The tau-function can be represented as an integral over Hermitian matrices
\be\label{matrixiin}
\tau_n\left(\left[\Lambda^{-1} \right]\right)=\frac{1}{ {\mathcal P}}\int_{{\mathcal H}_N} \left[d \mu ({Y})\right] e^{\Tr\left(Y\Lambda -e^Y+\tilde{q}e^{-Y}+\left(N/2-n\right)Y\right)}.
\ee
Here
\be\label{measyr1}
\left[d \mu (Y)\right]:=\frac{1}{(2\pi)^\frac{N}{2}\prod_{k=1}^N k!}\Delta(y)\Delta(e^{y})\left[d { U}\right]\prod_{i=1}^N e^{-\frac{N-1}{2}y_i}d y_i
\ee
is a non-flat measure on the space of Hermitian matrices. This measure is invariant with respect to the shift of the Hermitian matrix $Y$ by a scalar matrix
\be
\left[d \mu (Y)\right]=\left[d \mu (Y+cI)\right].
\ee
Let us shift the integration matrix $Y$ by $-I \log{\hbar}$. Then from (\ref{matrixiin}) we get (\ref{Int1}). This completes the proof.
\end{proof}
The measure (\ref{measyr1}) appears at the matrix models for the description of the Hurwitz numbers and other models of 
random partitions \cite{MSh,MMRP}. 

\begin{remark}
The measure $\left[d \mu ({Y})\right]$ is related to the flat measure $\left[dY\right]$ on the space of Hermitian matrices, 
\be
\left[d \mu ({Y})\right]=\exp\left(\sum_{i,j=0,i+j>0}\frac{(-1)^j}{2(i+j)}\frac{B_{i+j}}{i!j!}\Tr { Y}^i\Tr { Y}^j \right)\left[d{ Y}\right]\\
=\sqrt{\det\frac{\sinh\left(\frac{{Y}\otimes {I}-{I}\otimes{Y}}{2}\right)}{\left(\frac{{Y}\otimes {I}-{I}\otimes{Y}}{2}\right)}} \left[d {Y}\right].
\label{munorm}
\ee
Here the flat measure $\left[dY\right]$ is normalized by $\int_{{\mathcal H}_N} \left[dY\right]e^{-\frac{1}{2}\Tr Y^2}=1$.
Thus, we can represent (\ref{tauT}) as a Hermitian matrix integral with double-trace potential
\be
\tau_n\left(\left[\Lambda^{-1}\right]\right)=\frac{1}{ {\mathcal P}}\int_{{\mathcal H}_N} \left[d Y\right] e^{\Tr\left(Y\Lambda -e^Y+\tilde{q}e^{-Y}+\left(N/2-n\right)Y\right)+\sum_{i+j>0}\frac{(-1)^j}{2(i+j)}\frac{B_{i+j}}{i!j!}\Tr { Y}^i\Tr { Y}^j }.
\ee
\end{remark}
Let us stress that the size of the matrix $N$ is independent of the genus expansion parameter $\hbar$.

At least two other matrix models for the Gromov-Witten invariants of ${\bf P}^1$ were discussed in the literature. The Eguchi-Yang model \cite{Eguchi1,Eguchi2r} is not of the generalized Kontsevich type, and its relation to our model is not clear. A model, similar to (\ref{Int1}) was conjectured in \cite{Marino} in the context of topological string theory. In particular, the potentials of two models are close to each other, however, some important details are different. Let us also mention a model, which is a combination of the sum over partitions and matrix integral derived for the equivariant setup by Nekrasov \cite{Nikita}.

 \begin{remark}
After the first version of this paper was posted on the arXiv, we learned from G. Ruzza that a very similar matrix integral was discussed independently in his PhD thesis (see Section 7.3 of \cite{Ruzza}). The formula there was derived using the expressions for the correlation functions of the stationary intersection numbers obtained by Dubrovin, Yang and Zagier \cite{DYZ}.
\end{remark}

We also have an alternative matrix integral expression, based on (\ref{logint})
\be\label{matint2}
\tau({\bf t}^\omega,t_0^1)=\frac{e^{\frac{1}{\hbar}\Tr\left((t_0^1-\Lambda)\log\Lambda+\Lambda\right)+\frac{q}{\hbar^2}}}{\hbar^\frac{N^2}{2}}\int_{{\mathcal H}_N^{>0}} \left[d \mu_* ({X})\right] e^{\frac{1}{\hbar}\Tr\left(\left(\Lambda-\hbar/2-t^1_0\right)\log X -X+ \frac{q}{X}\right)},
\ee
where one integrates over positively defined Hermitian matrices with the measure
\be
 \left[d \mu_* ({X})\right]:=\frac{1}{(2\pi)^\frac{N}{2}\prod_{k=1}^N k!}\Delta(x)\Delta(\log(x))\left[d {U}\right]\prod_{i=1}^Nd x_i.
\ee

\section{Relative Gromov-Witten invariants of ${\bf P}^1$}\label{S4}

In this section, following \cite{OP2} we consider the infinite wedge/free fermions representation of the stationary Gromov-Witten theory of ${\mathbf P}^1$ relative to a point. Corresponding points of the Sato Grassmannian and matrix integrals are simple deformations of the points for the absolute model constructed in the previous section.

The connected stationary Gromov-Witten invariants of ${\mathbf P}^1$ relative to $m$ distinct points are
\be
\left<\prod_{i=1}^n \tau_{k_i}(\omega),\eta^1,\dots,\eta^m\right>^{{\mathbf P}^1}_{g,d}=\int_{\left[\overline{\mathcal M}_{g,n}({{{\mathbf P}^1},\eta^1,\dots,\eta^m})\right]^{vir}} \prod_{i=1}^n \psi_i^{k_i} \ev_i^*(\omega),
\ee 
where $\eta_1,\dots,\eta_m$ are the partitions of $d$, and $\overline{\mathcal M}_{g,n}({{{\mathbf P}^1},\eta^1,\dots,\eta^m})$  is the moduli space of the genus $g$, $n$-pointed relative stable maps with monodromy $\eta_i$ at the given point $q_i \in {\mathbf P}^1$. 

The invariants relative to two points, $0,\infty\in {\mathbf P}^1$ are denoted by
\be
\left<\mu,\prod_{i=1}^n \tau_{k_i}(\omega),\nu\right>^{{\mathbf P}^1}_{g,d},
\ee
where $\mu$ and $\nu$ are two partitions of $d$.
They reduce to the invariants, relative to one point if the partition $\mu$  is trivial, that is 
\be
\mu= (1^d),
\ee
and to the standard stationary theory, considered in the previous section, if both partitions $\mu$ and $\nu$ are trivial, i. e., $\mu=\nu= (1^d)$.

According to Okounkov and Pandharipande \cite{OP2}, the specialization of the generating function of the Gromov-Witten invariants relative to two points to the extended stationary sector
 is given by the sum over partitions
\be\label{angf}
\tau_{{\mathbf P}^1}({\bf x},{\bf t},{\bf s},y_0):=
\exp\left(\sum_{|\mu|=|\nu|}\frac{\mu t_\mu \,\nu s_\nu}{\hbar^{2\ell(\mu)}}
\left<\mu,\exp\left(\hbar^{-1}y_0\tau_0(1)+\sum_{i=0}^\infty \hbar^{i-1}x_i \tau_i(\omega)\right)\nu\right>^{{\mathbf P}^1}\right),
\ee
where
\be
\mu\, t_\mu:=\prod_{i=1}^{{l}(\mu)} \mu_i t_{\mu_i},
\ee
and the summation is over all partitions, including the empty one.
To simplify the presentation, we do not introduce the degree parameter $q$  in this section. It can be easily recovered from the dimensional reasons \cite{OP2}.
The generating function has genus expansion 
\be
\log\left(\tau_{{\mathbf P}^1}({\bf x},{\bf t},{\bf s},y_0)\right)=\sum_{g=0}^\infty \hbar^{2g-2}{\mathcal F}_g({\bf x},{\bf t},{\bf s},y_0).
\ee 
Let 
\be
\tau_n({\bf t},{\bf s};{\bf {\tilde{t}}})=\lvacn e^{J_+({\bf t})} \exp\left(\sum_{k=1}^\infty \tilde{t}_k {\mathcal P}_{k} \right) e^{J_-({\bf s})}\rvacn.
\ee
This is a tau-function of the 2D Toda lattice hierarchy in the variables ${\bf s}$, ${\bf t}$ and $n$  and it describes \cite{OP2} the generating function (\ref{angf}),
\begin{proposition}[Okounkov, Pandharipande]
\be\label{relat2}
\tau_{{\mathbf P}^1}({\bf x},{\bf t},{\bf s}, y_0)=\left.\tau_{y_0/\hbar}({\bf t}/\hbar,{\bf s}/\hbar;{\bf {\tilde{t}}})\right|_{\tilde{t}_k=\frac{\hbar^{k-1}x_{k-1}}{k!}}.
\ee
\end{proposition}
\begin{remark}
This tau-function also describes double weighted Hurwitz numbers and belongs to the class of the hypergeometric tau-functions of 2D Toda lattice hierarchy \cite{OrlovSch}.
\end{remark}

\subsection{Gromov-Witten invariants relative to one point}

In this section we consider the Gromov-Witten invariants of ${\mathbf P}^1$ relative to one point. They correspond to the specialization of (\ref{relat2}). Let
\be\label{tauext}
\tau_n({\bf s};{\bf {\tilde{t}}}):=e^{-s_1} \left. \tau_n({\bf t},{\bf s};{\bf {\tilde{t}}}) \right|_{t_k=\delta_{k,1}}
\ee
so that
\be
\left.\tau_{{\mathbf P}^1}({\bf x},{\bf t},{\bf s}, y_0)\right|_{t_k=\delta_{k,1}}=e^{s_1}\,\left.\tau_{y_0/\hbar}({\bf s};{\bf {\tilde{t}}})\right|_{\tilde{t}_k=\frac{\hbar^{k-1}x_{k-1}}{k!};\, s_k=\hbar^{-k-1} s_k}.
\ee
Here we introduce the factor $e^{-s_1} $ to satisfy the normalization $\tau_n({\bf s};{\bf 0})=1$.

By construction, this is an MKP tau-function in ${\bf s}$ variables. But we also conclude that: 
\begin{proposition}(\cite{Orlov03,AMMN})
$\tau_n({\bf s};{\bf {\tilde{t}}})$, given by (\ref{tauext}) is a tau-function of 2D Toda lattice hierarchy in the variables ${\bf {\tilde{t}}}$, ${\bf s}$ and $n$.
\end{proposition}
\begin{proof}
From (\ref{Wconf}) it follows that (\ref{tauext}) can be represented as
\begin{equation}\label{Todarep}
\begin{split}
\tau_n({\bf s};{\bf {\tilde t}})&= e^{-s_1}\lvacn e^{J_1} \exp\left(\sum_{k=1}^\infty \tilde{t}_k {\mathcal P}_{k} \right) e^{J_-({\bf s})}\rvacn\\
&=e^{-s_1}\lvacn e^{J_+({\bf {\tilde t}})} W^{-1} e^{J_1} e^{J_-({\bf s})}\rvacn\\
&= \lvacn e^{J_+({\bf {\tilde t}})} W^{-1} e^{J_-({\bf s})}\rvacn.
\end{split}
\end{equation}
The last line is of the form (\ref{stancor}), which completes the proof.
\end{proof}
For $s_k=q\delta_{k,1}/\hbar^2$ this tau-function coincides with (\ref{taunnn}):
\be
\tau_n({\bf t})=\left.\tau_n({\bf t},{\bf s})\right|_{s_k=q\delta_{k,1}/\hbar^2}.
\ee
Let us stress that in (\ref{tauext}) there is no direct symmetry between ${\bf s}$ and ${\bf\tilde{t}}$ variables. In a certain sense they are dual to each other, because
\be
\tau_n({\bf s};{\bf {\tilde t}})=\lvacn e^{J_+({\bf s})} W^{\star-1} e^{J_-({\bf {\tilde t}})}\rvacn.
\ee

 \begin{remark}
For $s_{k}=\frac{\delta_{k,r}}{r \hbar^{r+2}}$ the generating function (\ref{tauext}) describes the orbifold Gromov-Witten theory of ${\bf P}^1$, see \cite{Chen} and references therein.
 \end{remark}

\begin{remark}
We expect  that the generating function for the stable elliptic (and, probably, higher genera) invariants can be described by an integral transform in $\bf{s}$ using the methods, developed in \cite{MMRP}.
\end{remark}


\subsection{Basis vectors and Kac-Schwarz operators}
Let
\be\label{chta}
\Phi_k^{(n,{\bf s})}(z)= \frac{z^{n-z}e^z}{\sqrt{2\pi}} \int_{\rr} e^{y(z+k-n-1/2)-e^y+\sum_{m=1}^\infty s_m e^{-my}} dy.
\ee
Then using Lemma \ref{lemmainttran} one immediately concludes that ${\mathcal W}_{n,{\bf s}}=\sppan_{\cc} \{\Phi_1^{(n,{\bf s})},\Phi_2^{(n,{\bf s})},\Phi_3^{(n,{\bf s})},\dots\} \in \rm{Gr}^{(0)}_+$ is a point of the Sato Grassmannian, corresponding to (\ref{tauext}) as a KP tau-function in the variables $\tilde{t}_k$'s.
Similarly to Proposition \ref{KSlemma} we have:
\begin{proposition}
The Kac-Schwarz operators
\begin{equation}
\begin{split}
b:&=\left(\frac{z}{z+1}\right)^{n-z} (z+1)e^{\p_z-1},\\
a:&=1+(n+1/2-z)b^{-1}+\sum_{k=1}^\infty k s_k b^{-k-1}
\end{split}
\end{equation}
satisfy the commutation relation $\left[a,b\right]=1$ and completely specify a point of the Sato Grassmannian, 
\begin{equation}
\begin{split}
a \cdot\Phi_k^{(n,{\bf s})}(z) &=(k-1)\Phi_{k-1}^{(n,{\bf s})}(z) \in {\mathcal W}_{n,{\bf s}},\\
b \cdot \Phi_k^{(n,{\bf s})}(z)&=\Phi_{k+1}^{(n,{\bf s})}(z) \in {\mathcal W}_{n,{\bf s}}.
\end{split}
\end{equation}
\end{proposition}

Let us check that (\ref{tauext}) satisfies the first equation of the 2D Toda lattice hierarchy at ${\bf {\tilde t}}={\bf 0}$.
Since
\be
\hbar^{m}\frac{\p}{\p s_m}\, \Phi_k^{(n,{\bf s})}(z)=b^{-m} \cdot \Phi_k^{(n,{\bf s})}(z),
\ee
we have an identity
\be
(a-1)b\cdot \Phi_k^{(n,{\bf s})}(z)=\left(n+\frac{1}{2}-z+\sum_{k=1}^\infty k s_k \frac{\p}{\p s_k}\right)\Phi_k^{(n,{\bf s})}(z).
\ee
Then from the properties of the Kac-Schwarz operator \cite{A1} it follows that the tau-function (\ref{Todarep}) satisfies
\be\label{string}
\left(\sum_{k=1}^\infty k s_k \frac{\p}{\p s_k} -\frac{\p}{\p {\tilde t}_1}\right) \tau_n({\bf s};{\bf {\tilde{t}}}) =c\, \tau_n({\bf s};{\bf {\tilde{t}}})
\ee
for some $c$ (possibly dependent on ${\bf s}$ and $n$). We can find this constant from the consideration of (\ref{string}) at ${\bf {\tilde t}}={\bf 0}$
\be\label{firstsim0}
\tau_n({\bf s};{\bf 0})=1.
\ee
From (\ref{tauext}) we have
\begin{equation}\label{firstsim1}
\begin{split}
\frac{\p}{\p {\tilde t}_1}\tau_n({\bf s};{\bf {\tilde{t}}})\Big|_{{\bf \tilde{t}}={\bf 0}}&= e^{-s_1} \lvacn e^{J_1}  {\mathcal P}_{1} e^{J_-({\bf s})}\rvacn\\
&=e^{-s_1} \lvacn ( {\mathcal P}_{1}+J_1)e^{J_1}  e^{J_-({\bf s})}\rvacn\\
&=\lvacn ( {\mathcal P}_{1}+J_1) e^{J_-({\bf s})}\rvacn\\
&=(c_1(n)+s_1),
\end{split}
\end{equation}
where $c_1(n)$ is given by (\ref{ck1}), so that 
\be
\left(\sum_{k=1}^\infty k s_k \frac{\p}{\p s_k} -\frac{\p}{\p {\tilde t}_1}+\frac{n^2}{2}+\frac{s_1}{\hbar}\right) \tau_n({\bf s};{\bf {\tilde{t}}}) =\frac{1}{24} \tau_n({\bf s};{\bf {\tilde{t}}}).
\ee
Using (\ref{firstsim0}) and (\ref{firstsim1}) it is easy to show that the first equation of the 2D Toda lattice hierarchy is true at ${\bf \tilde{t}}={0}$, that is
\be
\left.\left(\tau_n({\bf s};{\bf {\tilde{t}}})\frac{\p^2}{\p t_1 \p s_1}\tau_n({\bf s};{\bf {\tilde{t}}}) -\frac{\p}{\p t_1} \tau_n({\bf s};{\bf {\tilde{t}}})\frac{\p}{\p s_1}\tau_n({\bf s};{\bf {\tilde{t}}})\right)\right|_{{\bf {\tilde t}}={0}}=\tau_{n-1}({\bf s};{\bf 0})\tau_{n+1}({\bf s};{\bf 0}).
\ee

\subsection{Quantum spectral curve and matrix model}

Let us modify the basis vectors (\ref{chta}) by the non-stable contributions 
\be
\Psi_k(z,n):=e^{\frac{1}{\hbar}\left(z\log(z)-z\right)-n\log(z)}\, \Phi_k^{(n,{\bf s})}\left(\frac{z}{\hbar}\right).
\ee
The modified functions have the integral representation
\be
\Psi_k(z,n):=\frac{\hbar^{1/2-k}}{\sqrt{2\pi}}\int_{\rr} e^{\frac{1}{\hbar} \left(y(z+\hbar(k-n-1/2))-e^y+\sum_{k=1}^\infty s_k \hbar^{k+1} e^{-ky}\right)}dy
\ee
and satisfy the quantum spectral curve equations
\be\label{qscext}
\left(e^{\hbar \p_z}+\sum_{k=1}^\infty s_k \hbar^{k+1} e^{-k\hbar \p_z}-z+\hbar(n+\frac{1}{2} -k)\right)\Psi_k(z,n)=0.
\ee
The classical curve, corresponding to the stationary sector of Gromov-Witten theory relative to one point, is
 \be
 e^{y}+\sum_{k=1}^\infty s_k e^{-ky}=x.
 \ee

Similar to Theorem \ref{T_1} we have
\begin{theorem}\label{T2}
The stationary generating function of Gromov-Witten invariant of ${\bf P}^1$ relative to one point is given by the asymptotic expansion of the  matrix integral
\be
\left.\tau_{{\mathbf P}^1}({\bf x},{\bf t},{\bf s}, y_0)\right|_{x_k=\delta_{k,1}}=\frac{e^{\frac{1}{\hbar}\Tr\left((t_0^1-\Lambda)\log\Lambda+\Lambda\right)+\frac{s_1}{\hbar^2}}}{\hbar^\frac{N^2}{2}}\int_{{\mathcal H}_N} \left[d \mu ({Y})\right] e^{\frac{1}{\hbar}\Tr\left(Y\Lambda -e^Y+ \sum_{k=1}^\infty s_k e^{-kY}+\left(N\hbar/2-t^1_0\right)Y\right)},
\ee
where
\be
t_k= \hbar\, k!\, Tr \Lambda^{-k-1}.
\ee
\end{theorem}
The generalization of the alternative integral formula (\ref{matint2}) is also straightforward. 

 \begin{remark}
For $s_{k}=\frac{\delta_{k,r}}{r \hbar^{r+2}}$ the quantum spectral curve (\ref{qscext}) reduces the curve obtained in a dual paramitrization by Chen and Guo \cite{Chen}.
In this case matrix model describes the stationary sector of the orbifold Gromov-Witten theory of ${\bf P}[r]$.
 \end{remark}

\begin{remark}
If we consider (\ref{tauext}) as a generating function of the weighted Hurwitz numbers \cite{ACEH,Harnad}, then it will correspond to the weight function 
\be
G(x)=\prod_{j=1}^N\left(\frac{\lambda_j-x+\hbar/2}{\lambda_j-x-\hbar/2}\right)^{\hbar^{-1}}.
\ee
Here $\lambda_j$,  the eigenvalues of the external matrix $\Lambda$, are the parameters.
It explicitly depends on $\hbar$ ($=\beta$ in the notations of \cite{ACEH}), hence, we do not expect that the approach of \cite{ACEH} to the topological recursion for weighted Hurwitz numbers will work in this case.
\end{remark}

\section*{Acknowledgments}
We thank A. Mironov and G. Ruzza
for useful discussions and correspondence and S. Shadrin for careful reading of the manuscript.
The author thanks an anonymous referee for  suggestions.
This work was supported by IBS-R003-D1 and by RFBR grant 18-01-00926.

\end{document}